\documentclass[12pt]{iopart}

\expandafter\let\csname equation*\endcsname\relax
\expandafter\let\csname endequation*\endcsname\relax
\usepackage{amssymb}
\usepackage{amsfonts}
\usepackage{amsmath}
\usepackage{mathtools}
\usepackage{amsthm}
\usepackage{subcaption} % packages for sub-figures
\usepackage{color}
\usepackage[normalem]{ulem}

\newtheorem{theorem}{Theorem}

\begin{document}

\title{Unique superdiffusion induced by directionality in multiplex networks}

\author{Xiangrong Wang,$^{1,2}$ Alejandro Tejedor,$^{3,4}$ Yi Wang,$^{1,2}$ Yamir Moreno$^{5,6,7\ast}$}

\address{$^{1}$Institute of Future Networks, Southern University of Science and Technology, Shenzhen, China}
\address{$^{2}$Research Center of Networks and Communications, Peng Cheng Laboratory, Shenzhen, China}
\address{$^{3}$Department of Science and Engineering, Sorbonne University Abu Dhabi, Abu Dhabi, United Arab Emirates}
\address{$^{4}$Department of Civil and Environmental Engineering, University of California Irvine, Irvine, CA, USA}
\address{$^{5}$Institute for Biocomputation and Physics of Complex Systems,University of Zaragoza, Zaragoza, Spain}
\address{$^{6}$Department of Theoretical Physics, University of Zaragoza, Zaragoza, Spain}
\address{$^{7}$ISI Foundation, Turin, Italy}
\ead{yamir.moreno@gmail.com}

%\vspace{10pt}
\begin{indented}
\item[] 09 September 2020
\end{indented}

\begin{abstract}
The multilayer network framework has served to describe and  uncover a number of novel and unforeseen physical behaviors and regimes in interacting complex systems. However, the majority of existing studies are built on undirected multilayer networks while most complex systems in nature exhibit directed interactions. Here, we propose a framework to analyze diffusive dynamics on multilayer networks consisting of at least one directed layers. We {\color{black} rigorously} demonstrate that directionality in multilayer networks can fundamentally change the behavior of diffusive dynamics: from monotonic (in undirected systems) to non-monotonic diffusion with respect to the interlayer coupling strength. Moreover, for certain multilayer network configurations, the directionality can induce a unique superdiffusion regime for intermediate values of the interlayer coupling, wherein the diffusion is even faster than that corresponding to the theoretical limit for undirected systems, i.e., the diffusion in the integrated network obtained from the aggregation of each layer. We theoretically and numerically show that the existence of superdiffusion is fully determined by the directionality of each layer and the topological overlap between layers. We further provide a formulation of multilayer networks displaying superdiffusion. Our results highlight the significance of incorporating the interacting directionality in multilevel networked systems and provide a framework to analyze dynamical processes on interconnected complex systems with directionality.
\end{abstract}

%\maketitle

\section{Introduction}
Multilayer networks provide a proper mathematical representation of complex systems with different types of interactions, e.g., social interactions among individuals across different social platforms such as Facebook and Twitter, and enable the understanding of dynamics acting on or evolving within those systems \cite{kivela2014multilayer,aleta2019multilayer}. A number of dynamical processes are studied on the framework of multilayer networks: for example, innovation/information diffusion \cite{Myers2012}, the control of formation and task allocation of swarms \cite{vicsek1995novel,berman2009optimized,prorok2017impact}, synchronization of oscillators \cite{saber2003consensus}, the understanding of functional brain connectivity \cite{abdelnour2014network}, and many others. More specifically, recent studies based on a multilayer network representation have revealed unforeseen dynamical regimes and behaviors, such as new spreading regimes when more than one diseases in spreading in a given population \cite{de2016physics,de2018fundamentals,sahneh2015exact}, and enhanced the stability of synchronized states \cite{del2016synchronization}. Moreover, conversely from building up a multilayer structure, gradually dismantling a multilayer network triggers an abrupt transition \cite{cozzo2019layer}, suggesting a non-additive effect of a multilayer structure from the integration of its layers. In terms of diffusion dynamics, multilayer network studies revealed that for any multilayer configuration, the diffusive behavior of the overall system (multilayer network) could be tuned to be faster than that of the slowest layer for significant values of interlayer diffusivity (in comparison to intralayer diffusivity). More surprisingly, but only for certain multilayer network configurations, the overall system can exhibit a superdiffusive behavior for large values of interlayer diffusivity, where the diffusive dynamics on the multilayer network is faster than in any of the individual layers, when those are considered in independently \cite{gomez2013diffusion,sole2013spectral,tejedor2018diffusion,cencetti2019diffusive}.

Most of the studies made in the field of multilayer networks adopt undirected structures, i.e., the interactions represented by both interlayer and intralayer links are assumed to be undirected. However, most real world systems are inherently directed. Examples are the World Wide Web, in which hyperlinks run in one direction from one Webpage to another; food webs, in which energy flows from prey to predator; phone call networks; metabolic networks; citation networks; social networks of followers and followee; gene regulation networks, to name a few. Additionally, some dynamical processes are governed by physical laws and inherently evolve in a directional manner despite of the underlying undirected connectivity topology, like current/information flow from high voltage to low voltage, virus transmissions from infected to susceptible individuals. Moreover, multi-leveled unidirectional interactions are commonly found in interacting systems. For example, a gene regulation network aims to properly characterize both protein-DNA and protein-protein directed interactions between genes and proteins \cite{hecker2009gene}. Therein, multilayer networks with directionality provide a realistic representation for such systems consisting of multilevel and directed interactions. 

Multilayer networks where {\color{black} at least one of the} layers is directed give rise to new theoretical challenges in analyzing their dynamical behaviors. As diffusion dynamics on undirected multilayer networks is fully characterized by the spectrum of network matrices, like supra-Laplacian \cite{gomez2013diffusion}, network directionality leads to nonsymmetrical graph matrices, which mostly features complex and non-orthogonal spectrum. Consequently, spectrum related analysis on undirected multilayer networks might be proven as invalid or inaccurate. Few recent studies on directed multilayer networks reveal new and unexpected physical behaviors \cite{tejedor2018diffusion,wang2019directionality}, as well as increased sensitivity to structural changes \cite{zhang2018altering}. 

The contribution of this paper is to provide a theoretical framework for diffusive dynamics on multilayer networks with directed layers. We show that although certain undirected multilayer networks can result in enhanced diffusive (faster than the slowest layer) or superdiffusive (faster than any of the layers) regimes for large values of the interlayer coupling, the diffusive rate of the overall system never exceeds the diffusion on the equivalent aggregated system (the direct sum of each layer). We also show that certain directed multilayer networks can also show enhanced diffusive and superdiffusive regimes for large values of interlayer coupling. Remarkably, directed multilayer networks, differently from their undirected counterparts, can exhibit a unique superdiffusive behaviour for intermediate values of coupling, where diffusion dynamics are even faster than the corresponding aggregated system. This unique superdiffusion regime only emerges when certain structural conditions are satisfied. We analytically and numerically uncover conditions driving the unique superdiffusion in directed multilayer networks. In addition, we propose a model to configure directed multilayer networks that achieves the unique superdiffusion with tunable magnitude.  

The paper is organized as follows. Diffusive dynamics on multilayer networks with direction is described in the Model section \ref{sec:model}. The section \ref{sec:results} of Results presents the conditions for diffusion regimes below the diffusivity in the corresponding aggregated system in section \ref{sec:subdiffusion} and for the superdiffusion regime above the diffusivity in the aggregated system in section \ref{sec:superdiffusion}, followed by a model for the construction of a multilayer network with superdiffusion in section \ref{sec:multilayer_design}. Section \ref{sec:conclusions} concludes the work.
\section{Model}\label{sec:model}
Multilayer networks consist of nodes interacting both within the same layer and across different layers via intralinks and interlinks which encodes multiple types of interactions. In this study, we focus on a particular type of multilayer networks called multiplex networks, characterized by layers consisting of the same set of nodes, but possibly different connectivity (layer topology); and layers interacting with each other only via counterpart node. The topological structure of the multiplex is described by the so-called supra-adjacency matrix  $ A = \left(a_{ij}\right) $, which is a block diagonal matrix. Each block is a $N \times N$ matrix, where $N$ is the number of nodes per layer. Each diagonal block corresponds to the adjacency matrix of a layer (intra-layer connectivity, and therefore an entry $ a_{ij}=1 $ in such a  block corresponds to an interlayer link), while the off-diagonal blocks, given the definition of a multiplex network, are just $N \times N$ identity matrices, representing the inter-layer links between replica nodes across layers. Note that a directed multiplex, is a multiplex where at least one of the diagonal blocks is not symmetric.
 
Let $ Q $ be the corresponding Laplacian matrix, defined as $ Q = D -A $ where $ D = \text{diag}\left(d_{i}\right) $ and $ d_i = \sum_{i=1}^{N}a_{ij} $ denoting the out-going degree of a node $ i $. We consider without loss of generality a multiplex network consisting of two directed layers and interconnected via links of weight $ p \geq 0 $, the corresponding Laplacian matrix $ Q $ can be written in a block form as
\begin{equation}\label{eq:def_Laplacian}
Q = \begin{bmatrix}
Q_1 +pI & -pI \\ -pI &Q_2 +pI \\
\end{bmatrix}
\end{equation}
where $ Q_i = D_i -A_i, \ i = 1, 2 $, denotes the Laplacian matrix for the directed graph in each layer and $ pI $ encodes the weighted interconnections across layers. The above defined Laplacian matrix for multilayer networks with directed layers is non-symmetrical, which might lead to complex rather than real spectrum. As the Laplacian satisfies a zero row sum, i.e., $ Qu=0 $ with $ u $ denoting the all one vector, value $ 0 $ is an eigenvalue with right eigenvector $ u $ and left eigenvector denoted as $ y $. For a nonsymmetric matrix, left and right eigenvectors are in general not the same.

To analyze the effect of network directionality on dynamical processes, we consider the standard diffusive dynamics on multilayer networks with directed layers and determine the convergence rate to a fixed point solution. Let $ x_i(t) $ denotes the state of a node $ i $ at time $ t $. Changes of the state $ x_i(t) $ with respect to time is governed by $ \dot{x}_i(t) = -d_{i}x_i(t)+\sum_i a_{ij}x_i(t) $. The diffusive dynamics on a multilayer network can be written in a matrix form as
\begin{equation}
\frac{dx(t)^T}{dt} = -x(t)^TQ
\end{equation}
where $()^T$ denotes transposition, and $ x(t) $ is the vector encoding the state of all nodes at time $ t $. The solution of the above governing equation follows $ x(t)^T = x(0)^Te^{-Qt} $. The fixed point $ x^{*T} $ at which $\dot{x}^{*T} = 0 $ is calculated employing the conservation law $ x(0)^Tu = x(\infty)^Tu $, where all one vector $ u $ is the right eigenvector corresponding to eigenvalue $ 0 $ of the Laplacian matrix $ Q $. Let $ y $ denote the left eigenvector of $ Q $ corresponding to eigenvalue $ 0 $. Combining the conservation law $ x(0)^Tu=x^{*T}u $ and $ y^TQ=0 $ yields the solution of fixed point
\begin{equation}\label{eq:steady_state}
x^{*T} = y^T\frac{x(0)^Tu}{y^Tu}
\end{equation} 
We show in \ref{appendix:convergence} a convergence to the fixed point in Eq. (\ref{eq:steady_state}) is guaranteed for a connected directed graph. The convergence rate to the fixed point $ x^* $ is characterized by the second smallest real part \cite{tejedor2018diffusion}, denoted as $ \text{Re}\lambda_{2}(Q)  $, of the eigenspectrum of the Laplacian matrix. 
\subsection{General analysis for diffusive behavior}\label{sec:super_sub_diffusion}
To analyze the convergence behavior of diffusive dynamics on multiplex networks with directed layers, we study the spectrum of the governing Laplacian matrix and compare with the integrated system. The characteristic polynomial for the Laplacian matrix $ Q $ in Eq. (\ref{eq:def_Laplacian}) can be calculated by, applying the Schur complement theorem on the block Laplacian matrix, 
\begin{equation}\label{eq:characteristic_polynomial}
\det\left(\lambda I -Q\right) 
=\det\left( \left(\lambda I -\frac{Q_1+Q_2}{2}\right)\left((\lambda-2p) I -\frac{Q_1+Q_2}{2}\right)-\Delta Q \right)
\end{equation}
where $ 4 \Delta Q = \left(Q_1-Q_2\right)^2-2\left(Q_1Q_2-Q_2Q_1\right) $. The eigen-pair of eigenvalue $ 2p $ and eigenvector  $ \left(u, -u\right) $ holds for directed multilayer networks, which is previously found in undirected multilayered networks \cite{sahneh2015exact,wang2019structural}. The joint effect of the Laplacian matrices $ Q_1 $ and $ Q_2 $ for each layer, encoded in the matrix $\Delta Q$, determines the deviation of the convergence rate of the whole system from that of the integrated multilayer. 

Depending on relations between the diffusive rate of the system as a whole and that of the integrated system, we refer two complementary regimes, superdiffusive and non-superdiffusive, as follows: (i) if $ \text{Re}\lambda_{2}(Q) $ of the overall system is greater than $ \lambda_{2}((Q_1+Q_2)/2) $ of the integrated system, we refer to as superdiffusive, (ii) otherwise as non-superdiffusive.
\begin{equation}\label{eq:def_superdiffusio}
\text{Re}\lambda_{2}(Q) \begin{cases}
> \lambda_{2}((Q_1+Q_2)/2) & \text{Superdiffusive}\\ 
\leq \lambda_{2}((Q_1+Q_2)/2) & \text{Non-superdiffusive}
\end{cases}
\end{equation}
We mention that the superdiffusive regime defined in this work slightly  differs  from  the definition in literature \cite{gomez2013diffusion,tejedor2018diffusion}, where superdiffusion implies a faster diffusion on the overall system than the diffusion on the fastest layer. Here, we mainly focus on diffusive behavior comparing the overall system and the correspondent integrated system, given the later serves as the asymptotic diffusion as intercoupling strength goes to infinity and, most importantly, serves as a theoretical upper limit for  diffusion on undirected multiplex networks (as shown in the following Eq. (\ref{eq:lambda2_relation_undirected})).

Under the special case of $ Q_1 = Q_2 $, solving Eq. (\ref{eq:characteristic_polynomial}) yields $ \text{Re}\lambda_{2}(Q) = \min\left(\lambda_{2}((Q_1+Q_2)/2), \ 2p \right) $ and therefore only the non-superdiffusive regime is exhibited. General cases of $  Q_1 \neq Q_2  $ result in possibilities of both non-superdiffusive and superdiffusive regimes. Figure \ref{fig:toy_model_super_sub_diffusion} exemplifies three typologies including (a) coupled directed and undirected layers, (b) coupled directed layers with the same direction, and (c) coupled directed layers with the opposite direction. Fig. \ref{fig:toy_model_super_sub_diffusion}a,c show only non-superdiffusive regimes and Fig. \ref{fig:toy_model_super_sub_diffusion}b exhibits superdiffusion.

Given the rich dynamical behavior of diffusion in directed multiplex, it is of paramount importance to identify the key structural properties underpinning each of the observed regimes. In this study, we mainly focus on identifying the underlying conditions driving the non-superdiffusive and superdiffusive regimes employing the principal submatrix approach and the Cauchy interlacing theorem \cite{cauchy1829equationa} (or the Poincare separation theorem). In addition, we employ the eigenvalue property of a normal matrix, which satisfies $ XX^T =X^TX $. For a normal matrix, the real parts of eigenvalues are the eigenvalues of the Hermitian part $ \text{Re}X = (X+X^T)/2 $ of the matrix $ X $. 

The Laplacian matrix is similar to the following matrix $ \widetilde{Q} $ in a transformed basis
\begin{equation}
\widetilde{Q} = \frac{1}{\sqrt{2}}\begin{bmatrix}
I &  I \\ -I &  I \\
\end{bmatrix}Q\begin{bmatrix}
I &  -I \\ I &  I \\
\end{bmatrix}\frac{1}{\sqrt{2}}
\end{equation}
which can be simplified as 
\begin{equation}
\widetilde{Q} =  \begin{bmatrix}
\frac{Q_1+Q_2}{2} & \frac{Q_2-Q_1}{2}\\
\frac{Q_2-Q_1}{2} & \frac{Q_1+Q_2}{2}+2pI
\end{bmatrix}
\end{equation}
The integrated Laplacian $ \frac{Q_1+Q_2}{2}  $ appears as a principal submatrix of the transformed Laplacian $ \widetilde{Q}  $ of the whole system. In addition, eigenvalues of the matrix $ \widetilde{Q} $ are the same with eigenvalues of $ Q $ due to the similarity relation between $ Q $ and $ \widetilde{Q} $. Applying the Cauchy interlacing theorem on the transformed matrix $ \widetilde{Q} $ and the integrated system $\frac{Q_1+Q_2}{2} $ as a principal submatrix, undirected multilayer networks always satisfy the following interlacing relation
\begin{equation}\label{eq:lambda2_relation_undirected}
\lambda_2\left(Q\right)\leq \lambda_2\left(\frac{Q_1+Q_2}{2}\right)
\end{equation}
Hence, there never occurs a superdiffusion in undirected multilayer networks, regardless of the strength of intercoupling. Eq. (\ref{eq:lambda2_relation_undirected}) serves as a theoretical limit and implies that although multilayer structure of undirected layers enhances the less diffusive layer, the diffusive rate is upper bounded by the integrated system. 

\section{Results}\label{sec:results}

\subsection{Conditions for the non-superdiffusive regime\label{sec:subdiffusion}}
For multilayer networks with directed layers, the conditions for non-superdiffusive and superdiffusive regimes are translated into conditions when the Cauchy interlacing theorem between a non-symmetric matrix and its principle submatrix is satisfied. Starting from the special case of coupled identical layers, it follows
\begin{equation}
\text{Re}\lambda_2\left(Q\right)=\lambda_2\left(\frac{Q_1+Q_2}{2}\right)
\end{equation}
which implies a non-superdiffusive regime for directed multilayer networks with identical directed layers for $ p \geq \lambda_2\left((Q_1+Q_2)/2\right) /2$. When the transformed Laplacian $ \widetilde{Q}$ and all the principal matrices including the integrated matrix $ (Q_2+Q_1)/2 $ are normal, the interlacing is also satisfied. Because for a principally normal matrix, all eigenvalues of matrices $ Q $ and $ (Q_2+Q_1)/2 $ lie on the same complex line with an angle $ \theta $ in the complex plane \cite{thompson1968principal,sherman2013principally}. For a principally normal matrix, the partial order of generalized interlacing on a complex plain is obtained:
\begin{equation}
|\lambda_2\left(Q\right)| \leq \left|\lambda_2\left(\frac{Q_1+Q_2}{2}\right)\right|
\end{equation}
As the real part of eigenvalue reads $ |\lambda_2\left(Q\right)| e^{i\theta} $ and $ \theta $ is the same for all eigenvalues, it implies an interlacing between real parts of eigenvalues and accordingly a non-superdiffusive regime for multilayer networks with principally normal Laplacian matrices. 

When the requirement of normality of all principal submatrices is relaxed to the normality of a single principal submatrix, generalized interlacing of lexicographic order can be applied to $ \widetilde{Q} $ and $  (Q_2+Q_1)/2 $. For a normal Laplacian matrix and a normal principal submatrix $ \frac{Q_1+Q_2}{2} $, we have that
\begin{equation}\label{eq:lexicographic_interlacing}
\lambda_2\left(Q\right)\leq_\theta \lambda_2\left(\frac{Q_1+Q_2}{2}\right)
\end{equation}
where $ \leq_\theta $ represents the lexicographic order \cite{fan1957imbedding} characterized by the positive cone $ H \coloneqq \{a+bi \colon a > 0 \ \text{or} \ a=0 \ \text{and} \ b >0\} $. Eq. (\ref{eq:lexicographic_interlacing}) means that $e^{-i\theta}\lambda_2\left(\frac{Q_1+Q_2}{2}\right)-e^{-i\theta}\lambda_2\left(Q\right)  $ lies in a positive cone, which implies $ \lambda_2\left(\frac{Q_1+Q_2}{2}\right) \geq \text{Re}\lambda_2\left(Q\right) $. Therefore, directed multilayer networks with normal Laplacian matrix and normal principal submatrix $ \frac{Q_1+Q_2}{2} $ exhibit only the non-superdiffusive regime.

\subsection{Conditions for the superdiffusive regime \label{sec:superdiffusion}}
Diffusion on undirected multilayer never exceeds the diffusion of integrated part. However, certain multilayer networks with directed layers break the constraint and exhibit a unique superdiffusion. In this subsection, we analyze conditions for the occurrence of such superdiffusion. For the generalized interlacing theorem to hold, it requires both the whole system and the integrated system to be normal or simultaneously close to normal. When the condition is relaxed such that the integrated system is normal or close to normal, while the whole system deviates from normality, it might occur with superdiffusion. We provide analytical evidences by firstly relating the real part of eigenvalues and eigenvalues of the Hermitian part of the original matrix, which follows 
\begin{equation}
\sum_{j=N-k}^{N}\text{Re}\lambda_{j} = \sum_{j=N-k}^{N}x_j^*\text{Re}\widetilde{Q}x_j \leq \sum_{j=N-k}^{N}\lambda_{j}\left(\text{Re}\widetilde{Q}\right)
\end{equation}
where $ \text{Re}\widetilde{Q}$ denotes the Hermitian part of $\widetilde{Q}$. Because $ \sum_{j=1}^{N}\text{Re}\lambda_{j} = \sum_{j=1}^{N}\lambda_{j}\left(\text{Re}\widetilde{Q}\right)$ and $ \lambda_{1}\left(\text{Re}\widetilde{Q}\right)=0 $ for a normal subgraph, we have that $ \text{Re}\lambda_{2}\left(Q\right) \geq  \lambda_{2}\left(\text{Re}\widetilde{Q}\right) $. The equality is achieved if the matrix is normal. Therefore, if the integrated system is normal, it implies $ \lambda_{2}\left(\frac{Q_1+Q_2}{2}\right) = \lambda_{2}\left(\text{Re}\frac{Q_1+Q_2}{2}\right) $. 
 
Secondly, we analyze the superdiffusion condition by relating eigenvalues of the Hermitian part of the integrated system and Hermitian part of the whole system. On one hand, observing that the Hermitian part of the integrated system coincides with the principal submatrix of the Hermitian part of the whole system, we have that $ \lambda_{2}\left(\text{Re}\widetilde{Q}\right) \leq \lambda_{2}\left(\text{Re}\frac{Q_1+Q_2}{2}\right)  $ due to the Cauchy interlacing theorem for Hermitian matrices. On the other hand, $ \lambda_{2}\left(\text{Re}\widetilde{Q}\right) \geq \lambda_{2}\left(\text{Re}\frac{Q_1+Q_2}{2}\right) + \lambda_{2}\left(\text{Re}\frac{Q_2-Q_1}{2}\right)  $ due to Ky Fan majorization \cite{Fan1950On} theorem $ \lambda\left(A+B\right) \prec \lambda\left(A\right)+\lambda\left(B\right) $. Therefore, if the difference $ \lVert \text{Re}\left(Q_2-Q_1\right) \rVert_2$ is bounded, it leads to $ \lambda_{2}\left(\text{Re}\widetilde{Q}\right) \approx \lambda_{2}\left(\text{Re}\frac{Q_1+Q_2}{2}\right) $.

A superdiffusion is thus established
\begin{equation}
\text{Re}\lambda_{2}\left(Q\right) \gtrapprox  \lambda_{2} \left(\frac{Q_1+Q_2}{2}\right)
\end{equation}
if (i) the integrated system is or close to normal and (ii) the structure difference of coupled layers $ \text{Re}Q_2-\text{Re}Q_1$ is negligibly small.

The equality of $ \lambda_{2}\left(\text{Re}\widetilde{Q}\right) = \lambda_{2}\left(\text{Re}\frac{Q_1+Q_2}{2}\right) $ is reached for a multilayer network consisting of a directed graph coupled with its reversed graph, $ Q_2=Q_1^T $. In this case, it follows $  \text{Re}Q_2 = \text{Re}Q_1$ and there exists a superdiffusion of $ \text{Re}\lambda_{2}\left(Q\right) >  \lambda_{2} \left(\frac{Q_1+Q_2}{2}\right) $.

It was reported that the superdiffusion occurs when a fully connected network coupled with a network with a certain level of directionality, quantified by a metric called network directionality index \cite{tejedor2018diffusion}. Here, we prove in \ref{appendix:NDI} an open problem of the normative property of the Network directionality index (NDI). In addition, we show that the condition of sufficient directionality as quantified by NDI is in line with the derived normality conditions in this study (as shown in Figure \ref{fig:NDI_nonnormality}). The derived condition for superdiffusion in this work, i.e., a close to normal and a deviation from normal (upper left panel in Fig. \ref{fig:NDI_nonnormality}) for the integrated and overall multiplex, respectively, is equivalent to the condition of a sufficient level of NDI. The nonnormality condition for the non-superdiffusive regime (lower right panel in Fig. \ref{fig:NDI_nonnormality}) also agrees with the NDI condition. However, the nornormality condition is widely applicable for multiplex networks with any number of directed layers, compared to a single directed layer dealt by NDI condition.

\subsection{Formulation of multilayer networks with superdiffusion\label{sec:multilayer_design}}
To interpret the normality conditions for superdiffusion, we propose a model to construct multiplayer networks exhibiting superdiffusion with a tunable level of magnitude. When the coupled layers have identical Hermitian part of Laplacian, the minimization of the normality level of the integrated structure is translated into
\begin{equation}\label{eq:min_normality_integrated}
\min ||Q_1Q_2^T - Q_2^T Q_1||_2
\end{equation}

with proof in \ref{appendix:minnormality}. Achieving superdiffusion by minimizing Eq. ({\ref{eq:min_normality_integrated}}) is in general dauntingly difficult, even for undirected graphs of symmetrical Laplacian matrices. The simultaneous diagonalization, which always leads to the commutativity of matrices $ Q_1 $ and $ Q_2^T $, i.e., $ Q_1Q_2^T = Q_2^T Q_1 $, was posted as an open problem by Hiriart-Urruty \cite{hiriart2007potpourri} more than a decade ago. Extensive research is performed due to the close association with quadratically constrained quadratic programming and certain advances are made for real symmetrical matrices \cite{jiang2016simultaneous}. However, results on nonsymmetrical matrices are extremely scarce. Nonetheless, we propose a construction model to generate a directed multilayer network by constructing one matrix (e.g., $ Q_2 $) from a given matrix (e.g., $ Q_1 $), such that superdiffusion occurs.

We decompose the construction model into two parts: (i) constructing a graph $ G_2 $ from the replication of graph $ G_1 $ but after reversing the direction of directed links; (ii) preserving the out-degree of each node in graph $ G_1 $ such that graph Laplacian $ Q_2 $ has the same diagonal elements as $ Q_1 $. To this end, we propose to reverse links in a cycle fashion, i.e., simultaneously reversing the direction of all links involved in a cycle. In addition, the number of links reversed in a cycle fashion is closely associated with the magnitude $ \text{Re}\lambda_{2}\left(Q\right)-\lambda_{2} \left(\frac{Q_1+Q_2}{2}\right) $ of superdiffusion. For two directed graphs $ G_1 $ and $ G_2 $, we define a variable $ q $ as direction overlap, the fraction of links having the same direction and calculated by $ q = \frac{\lvert \mathcal{E}(G_1) \cap \mathcal{E}(G_2)\rvert }{\lvert \mathcal{E}(G_1) \cup \mathcal{E}(G_2)\rvert } $ or in a matrix form as
\begin{equation}
q = \frac{\sum_{i,j}\left(A_1\right)_{ij}\left(A_2\right)_{ij}}{\sum_{i,j}\left(A_1\right)_{ij}+\sum_{i,j}\left(A_2\right)_{ij}-\sum_{i,j}\left(A_1\right)_{ij}\left(A_2\right)_{ij}}
\end{equation}
The denominator normalizes the direction overlap, in which $ q=0 $ corresponds to $ Q_2 = Q_1^T $ and $ q = 1 $ corresponds to $ Q_1 = Q_2 $. The magnitude of superdiffusion can thus be tuned by reversing different fraction of links, parameterized by $ 1-p $.

Figure \ref{fig:lower_bound_direction_overlap} shows diffusive behaviors on constructed multiplex networks consisting of (a) a circulant directed graph and the duplicate digraph with $ 1-q $ reversed links and (b) a real-world genetic and protein interactions network of the Epstein–Barr virus \cite{stark2006biogrid,de2015muxviz} and the duplicate digraph with $ 1-q $ reversed links. The coupled digraph and its duplicate without link reversing ($ q=1 $) displays a non-superdiffusive behavior. For $ 0 \leq q < 1$, the constructed multiplex built upon both synthetic and real-world directed graph shows superdiffusion regimes. In addition, a higher fraction of reversed links (a smaller $ q $) results in a higher magnitude of superdiffusion.

\section{Conclusions}
\label{sec:conclusions}
In this paper, we study the diffusive dynamics on multiplex networks consisting of at least one directed layers. Though multilayer structure of undirected layers enhances diffusion, it is restricted to a theoretical limit by the aggregation of each layer. We show that ubiquitous property of directionality in each layer could break this limit and achieve a unique superdiffusion regime, wherein diffusion is even faster than the corresponding aggregated system. We analytically and numerically uncover that the unique superdiffusion is driven by the directionality and the underlying structure of layers rather than how strongly layers are coupled. In particular, diffusive behaviors are associated with the non-normality level of the integrated and the overall system: when both the integrated and the overall system are normal, it is assured that the multiplex exhibits only a non-superdiffusive regime; when the integrated system is normal or close to normal while the whole system is deviated from normal, it exhibits the unique superdiffusive regime. Additionally, we provide a model to construct multiplex networks, achieving superdiffusion with a tunable level of magnitude. We show that directionality induces a unique superdiffusion, a regime not existed in the counterpart of undirected multilayer, and suggest the key role of network directionality in analyzing diffusive dynamics on real-world systems which oftentimes are structurally directed and multilayered.

\begin{figure}[!htp]%
	\centering
	\includegraphics[width=\textwidth]{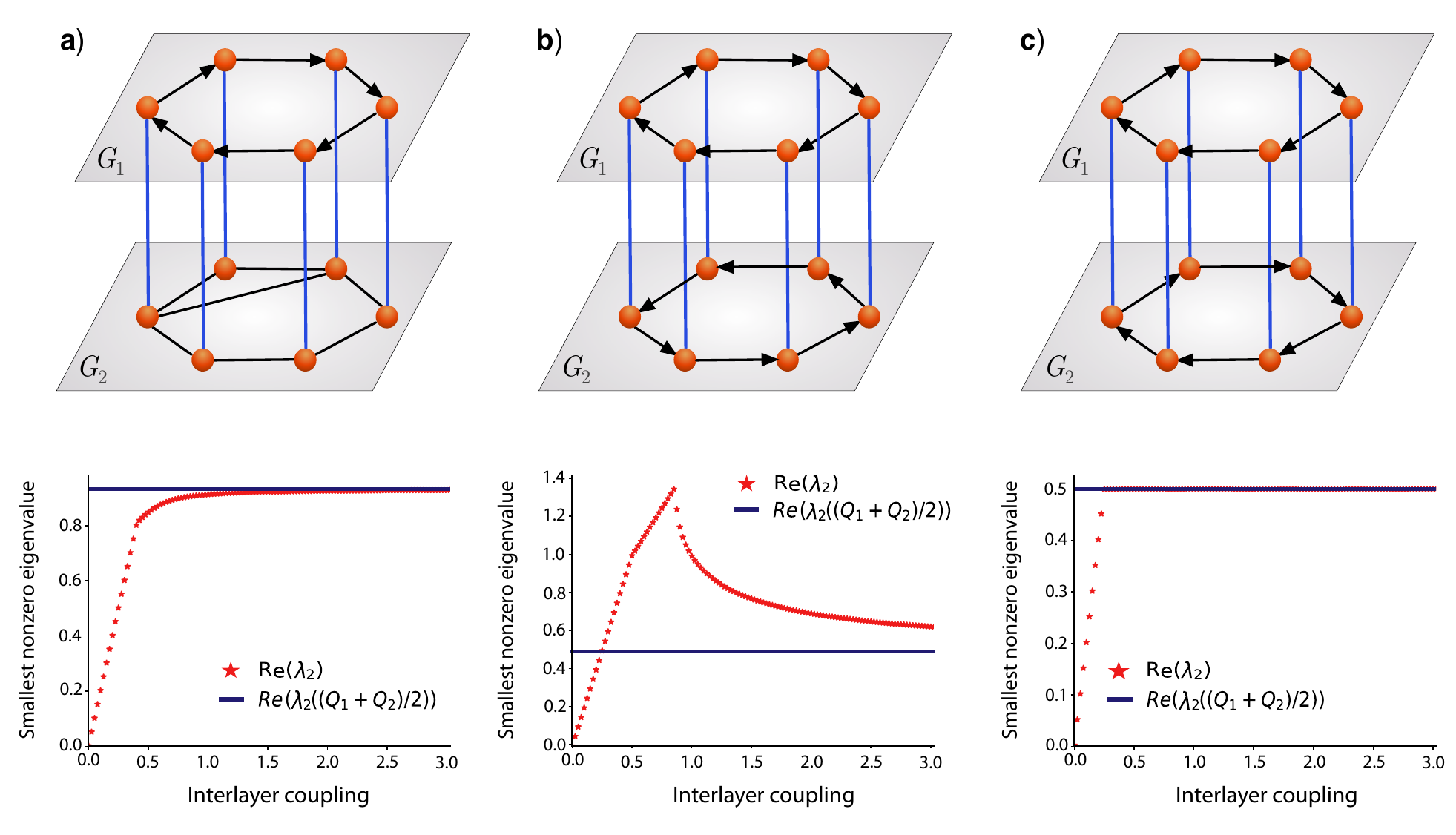}
	\caption{Superdiffusive and non-superdiffusive regimes on multiplex networks with at least one directed layers. Panel (a) shows a multiplex network consisting of a directed layer and an undirected layer displaying a non-superdiffusive regime $( \text{Re}\lambda_{2}\left(Q\right) < \lambda_{2}\left(\frac{Q_1+Q_2}{2}\right) )$.  Panel (b) shows a multiplex network with coupled layers in a reversed direction, exhibiting a superdiffusion $( \text{Re}\lambda_{2}\left(Q\right) > \lambda_{2}\left(\frac{Q_1+Q_2}{2}\right) )$ after the transition from a linear growth of $ 2p $. Panel (c) shows a multiplex with coupled layers in the same direction, exhibiting a non-superdiffusive regime $( \text{Re}\lambda_{2}\left(Q\right) =\lambda_{2}\left(\frac{Q_1+Q_2}{2}\right)) $. }
	\label{fig:toy_model_super_sub_diffusion}
\end{figure}
\begin{figure}[!htp]%
\centering
	\includegraphics[width=\textwidth]{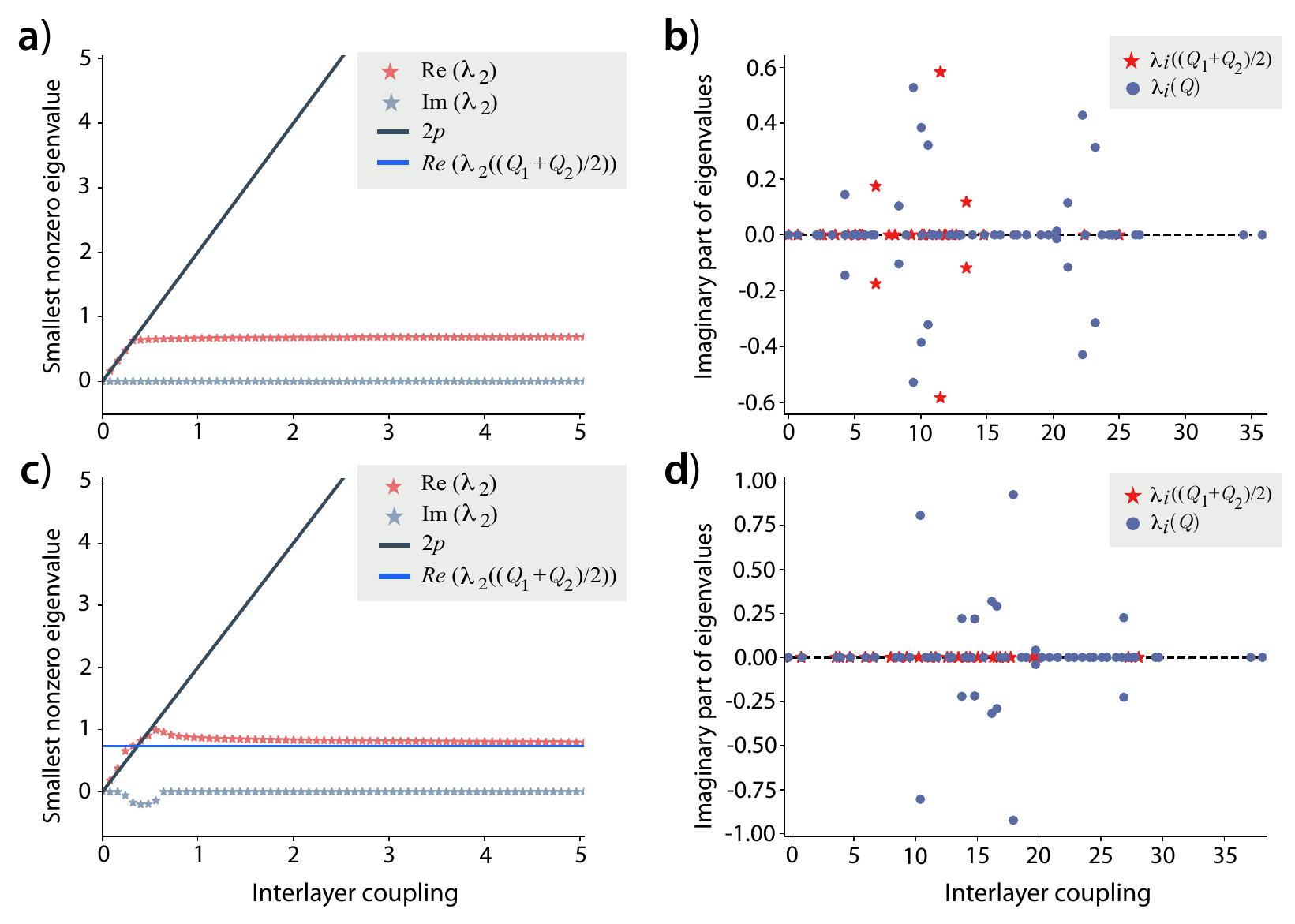}
\caption{Non-superdiffusive and superdiffusive real-world multiplex networks. Panel (a) shows a non-superdiffusive regime on the \emph{Vickers-Chan multiplex} social networks \cite{vickers1981representing} with directed layers. Panel (b) shows a superdiffusion on the multilayer network where one layer is from \emph{Vickers-Chan multiplex} and the other layer is the reverse of the first layer. Panels (b) and (d) show the distribution of all eigenvalues of the Laplacian $ Q $ for the overall multiplex and that of the integrated system $ \left(Q_1+Q_2\right)/2 $ corresponding to networks in panels (a) and (c), respectively.}
\label{fig:diffusion_realNets}
\end{figure}
\begin{figure}[!htp]
	\centering
	\begin{subfigure}[b]{0.495\textwidth}
		\includegraphics[width=\textwidth]{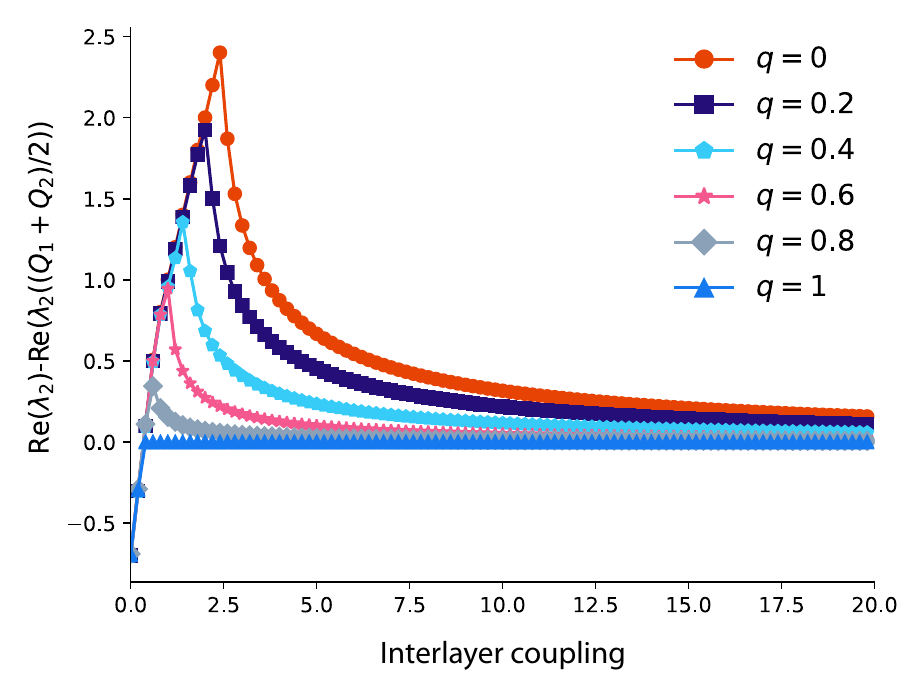}
	\end{subfigure}
	\begin{subfigure}[b]{0.495\textwidth}
		\includegraphics[width=\textwidth]{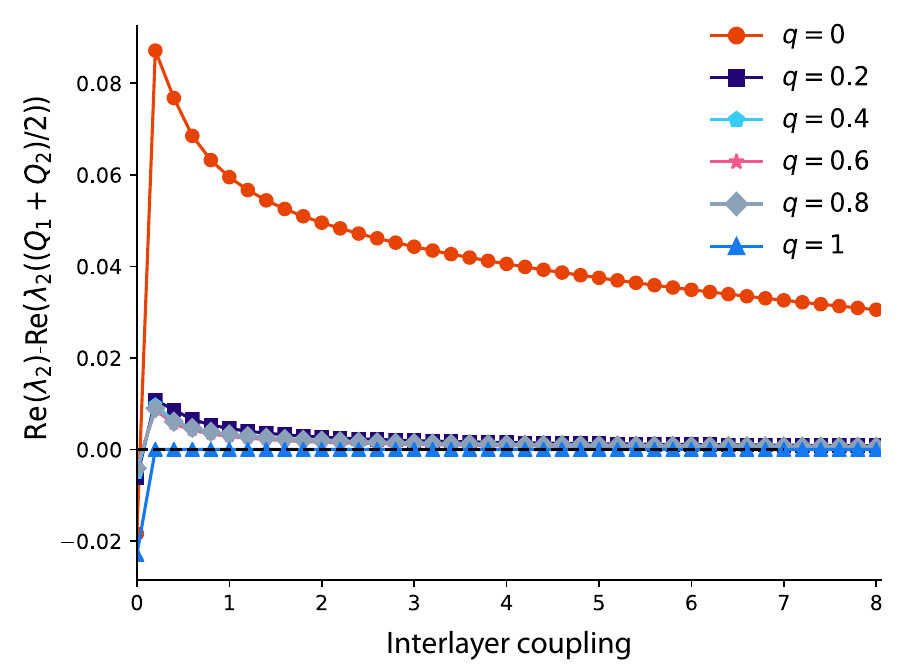}
	\end{subfigure}
	\caption{The difference $  \text{Re}\lambda_2(Q) - \lambda_2\left(\frac{Q_1+Q_2}{2}\right) $ as a function of the interlayer coupling strength $ p $ for (a) a multilayer network consisting of directed circulant graph of $ 100 $ nodes and average degree $ 4 $ and the constructed layer by the proposed model; and (b) a multilayer network consisting of a genetic and protein interactions network of the Epstein–Barr virus network and the constructed layer.}
	\label{fig:lower_bound_direction_overlap}
\end{figure}
\begin{figure}[!htp]
	\centering
	\includegraphics[width=0.6\textwidth]{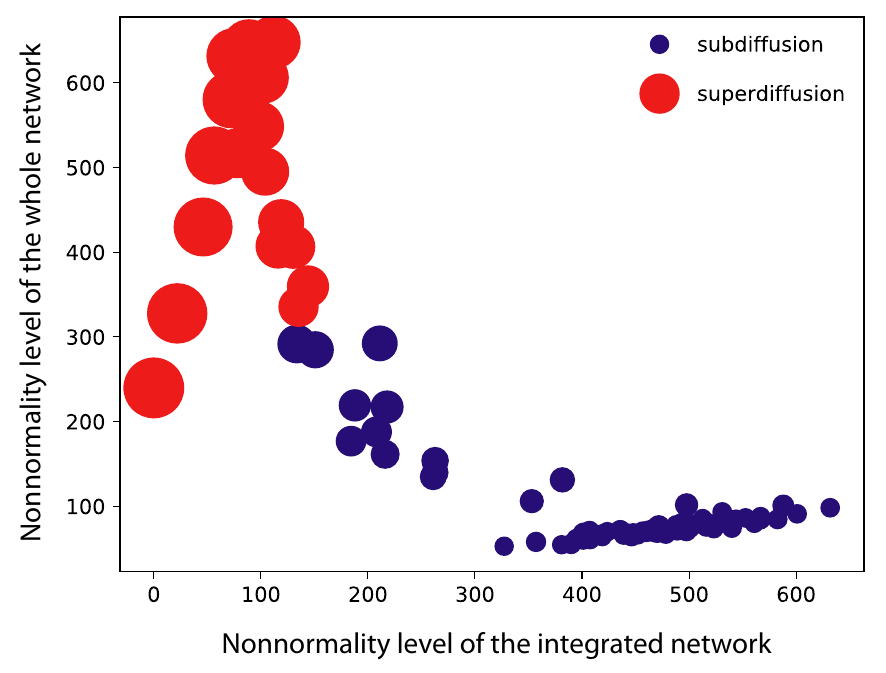}
	\caption{The normality condition is in line with the condition of sufficient directionality for the occurrence of superdiffusion. The size of each circle is proportional to the value of network directionality index, NDI. The simulation is performed on a two-layered network of $ N = 200 $ nodes wherein one layer consists of a fully connected, undirected graph and the other layer consists of a directed graph whose NDI is modified progressively. The intercoupling weight is $ p=200 $. A high level of NDI is translated to a low deviation from normality of the integrated network and a high deviation of normality of the multiplex network as a whole.}
	\label{fig:NDI_nonnormality}
\end{figure}

\appendix
\section{Convergence to the fixed point}\label{appendix:convergence}
The convergence to the fixed point is characterized by $ \Delta x(t)^T=(x(t)-x^*)^T $. The rate of convergence is regulated by $ \dot{\Delta x}(t)^T = \frac{d(x(t)-x^*)^T}{dt}=-(\Delta x(t)^T+x^{*T})Q=-\Delta x(t)^TQ $. The norm $ \lVert \Delta x(t)^T \rVert $ of the difference vector is bounded by
\begin{equation}
\lVert \Delta x(t)^T \rVert = \lVert \Delta x(0)^T e^{-Qt} \rVert  \leq \lVert \Delta x(0)^T \rVert \lVert e^{-Qt} \rVert 
\end{equation}
If one mainly concerns with the convergence behavior for the decrease of $ \lVert \Delta x(t) \rVert  $, the following holds for a bounded operator $ Q $ in a Hilbert space
\begin{equation}
\lim_{t \rightarrow \infty}\frac{\log \lVert e^{-Qt} \rVert}{t} =-\lambda_{\min}
\end{equation}
where $ \lambda_{\min} $ is the minimum eigenvalue of the Laplacian $ Q $. For a connected directed graph, $ \lambda_{\min} $ equals to zero which means a guaranteed convergence to the steady state $ x^* $ calculated by Eq. (\ref{eq:steady_state}). 
\section{Proof that network directionality index is normalized.}\label{appendix:NDI}
In the appendix, we prove that the metric quantifying the level of directionality is normalized to $ 1 $. Rewrite the definition \cite{tejedor2018diffusion} of Network Directionality Index (NDI) as
\begin{equation}
\text{NDI} = \frac{<\Delta d_{ij}>}{<d_{ij}>}
\end{equation}
or equivalently,
\begin{equation}\label{eq_def_NDI}
\text{NDI} = \frac{\sum_{i = 1}^{N}\sum_{j > i}\Delta d_{ij}}{\sum_{i = 1}^{N}\sum_{j > i }\frac{ d_{i \rightarrow j } + d_{j \rightarrow i}}{2}}
\end{equation}
where $ \Delta d_{ij} =  |d_{i \rightarrow j } - d_{j \rightarrow i}| $.
\begin{theorem}
	\label{theorem:NDI_simplification}
	For a directed cycle network consisting of $ N $ nodes and $ N $ directed links in a clockwise or counterclockwise direction, the $ \text{NDI} $ can be simplified as
	\begin{equation}
	\text{NDI} = \frac{2\sum_{j =1}^{N}|d_{i \rightarrow j } - d_{j \rightarrow i}|}{\sum_{j =1 }^{N}\left(d_{i \rightarrow j } + d_{j \rightarrow i}\right)}
	\end{equation}
	where $ \sum_{j =1}^{N}|d_{i \rightarrow j } - d_{j \rightarrow i}|/N $ is the average path asymmetry and $ \sum_{j =1 }^{N}\left(d_{i \rightarrow j } + d_{j \rightarrow i}\right)/2N $ is the average path length starting from an arbitrary node $ i $ and arriving at each of the remaining $ N -1 $ nodes.
\end{theorem}
\begin{proof}
	For a directed circulant graph of $ N $ nodes and $ N $ directed links in a clockwise direction, the shortest directed path from node $ i $ to node $ j $ and back to $ i $ forms a closed cycle of length $ N $, regardless of choices of pairs of nodes ($ i \neq j $). Thus, we have that 
	\begin{equation}\label{eq_num}
	\sum_{i = 1}^{N}\sum_{j > i } \frac{ d_{i \rightarrow j } + d_{j \rightarrow i}}{2} = \frac{N}{4} \sum_{j =1 }^{N}\left(d_{i \rightarrow j } + d_{j \rightarrow i}\right)
	\end{equation}
	Applying the symmetry of $ \Delta d_{ij} = \Delta d_{ji} $ into the numerator of (\ref{eq_def_NDI}), we have that
	\begin{equation}
	\sum_{i = 1}^{N}\sum_{j > i} \Delta d_{ij} = \frac{1}{2}\sum_{i = 1}^{N}\sum_{j =1}^{N} \Delta d_{ij} 
	\end{equation}
	For a directed cycle consisting of $ N $ nodes and $ N $ directed links in a clockwise direction, the sum $ \sum_{j =1}^{N} \Delta d_{ij} $ is the same for each node $ i $ and we arrive at
	\begin{equation}\label{eq_denom}
	\sum_{i = 1}^{N}\sum_{j = 1}^{N} \Delta d_{ij} = N \sum_{j =1}^{N} \Delta d_{ij}
	\end{equation}
	Substituting (\ref{eq_num}) and (\ref{eq_denom}) into the definition of $ \text{NDI} $ establishes Theorem \ref{theorem:NDI_simplification}.
\end{proof}
\begin{figure}[h]%
	\centering
	\includegraphics[width=0.4\textwidth]{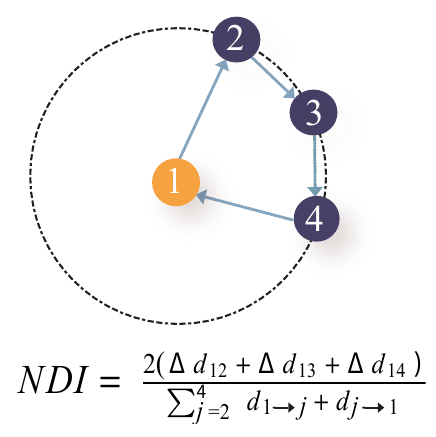}%
	\caption{A toy structure to exemplify the simplification of NDI by Theorem \ref{theorem:NDI_simplification}.}%
	\label{fig_def_simplification}%
\end{figure}
\begin{theorem}
	\label{theorem:add_link}
	Given a directed cycle consisting of $ N $ nodes and $ N $ directed links in a clockwise direction, adding a single or multiple directed links between any pair of nodes $ i $ and $ j $ generates shortcuts either in forward ($ i \rightarrow j $) or backward ($ j \rightarrow i $) direction between that pair of nodes and thus forms $ k $ ($ k \geq 1 $) sub-cycles with each consisting of $ c_k  $ ($ c_k \leq N $) nodes and $ c_k $ links. The $ \text{NDI}^\prime $ after adding directed links is smaller than or equal to the $ \text{NDI }$ of the original directed cycle. 
\end{theorem}
\begin{proof}
	Adding a directed link connecting node $ j $ to $ i $ forms a sub-cycle $ k $ consisting of $ c_k $ nodes and $ c_k $ directed links in a clockwise direction, as shown in Figure \ref{fig_add_link}. 
	\begin{figure}[h]%
		\centering
		\includegraphics[width=0.6\textwidth]{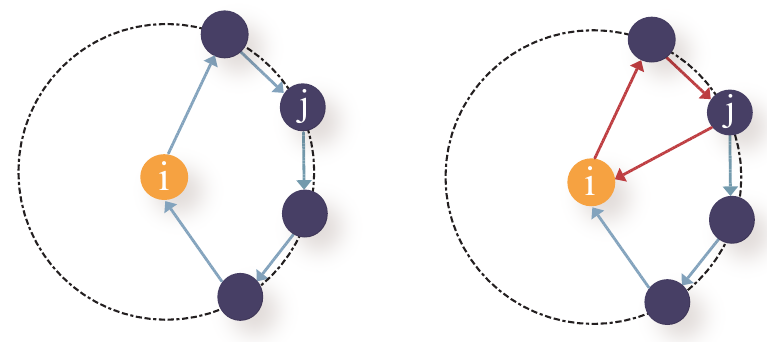}
		\caption{Adding a directed link $ j \rightarrow i $ forms a sub-cycle consisting of $ c_k $($ =3 $) nodes and $ c_k $ directed links in a clockwise direction.}%
		\label{fig_add_link}%
	\end{figure}
	For a pair of nodes $ i $ and $ j $ in the sub-cycle $ k $, the difference of path asymmetry $ \Delta d_{ij} $ before and after adding directed links reads 
	\begin{equation}
	\Delta d_{ij}^\prime -  \Delta d_{ij} = N-2\min\left(d_{i \rightarrow j},\ d_{j \rightarrow i}\right)
	- \left(c_k-2\min\left(d_{i \rightarrow j},\ d_{j \rightarrow i}\right)\right) = N-c_k
	\end{equation}
	and the difference of $ d_{i \rightarrow j } + d_{j \rightarrow i}  $ before and after link addition reads 
	\begin{equation}
	d_{i \rightarrow j }^\prime + d_{j \rightarrow i}^\prime - \left(d_{i \rightarrow j } + d_{j \rightarrow i}\right) = N - c_k
	\end{equation}
	Summing up all the node pairs starting from node $ i $ to the remaining $ N-1 $ nodes yields
	\begin{equation}
	\text{NDI}^\prime = \frac{2\sum_{j =1}^{N}\left(|d_{i \rightarrow j } - d_{j \rightarrow i}| - (N-\left(m\right)_j )\right)}{\sum_{j = 1}^{N}\left(\left(d_{i \rightarrow j } + d_{j \rightarrow i}\right)  -  (N-\left(m\right)_j ) \right)}
	\end{equation}
	where $ m $ is a column vector of length $ N $ with element $ \left(m\right)_j = c_k $ if node $ j $ belongs to a sub-cycle of $ c_k $ nodes and $ c_k $ directed links in a clockwise direction, which can be further simplified as 
	\begin{equation}
	\text{NDI}^\prime = \frac{2\sum_{j =1}^{N} |d_{i \rightarrow j } - d_{j \rightarrow i}| -\sum_{j}\left(N-\left(m\right)_j \right)}{\sum_{j = 1}^{N}\left(d_{i \rightarrow j } + d_{j \rightarrow i}\right)  -  \sum_{j}\left(N-\left(m\right)_j \right)}
	\end{equation}
	It has been shown \cite{tejedor2018diffusion} that a directed cycle of $ N $ nodes has $ \text{NDI} <1$.  Subtracting the same number $ \sum_{j}\left(N-\left(m\right)_j \right) $ both in the numerator and denominator of (\ref{eq_def_NDI}) does not increase the $ \text{NDI} $ of the original cycle consisting of $ N $ nodes.
\end{proof}
Since a directed cycle of $ N $ nodes has $ \text{NDI} \leq 1$ (Theorem \ref{theorem:NDI_simplification}) and adding directed links does not increase the $ \text{NDI} $ (Theorem \ref{theorem:add_link}), we prove that $ \text{NDI} \leq 1 $.

\section{Proof of Eq. (\ref{eq:min_normality_integrated})}\label{appendix:minnormality}
Given the equality of the Hermitian part of the coupled directed layers, we show that the close to normal condition of the integrated system is translated to the minimization of Eq. $\left(\ref{eq:min_normality_integrated}\right)$ in the main text. The definition of normality of the integrated system reads
$ \left(Q_1+Q_2\right)\left(Q_1+Q_2\right)^T = \left(Q_1+Q_2\right)^T\left(Q_1+Q_2\right) $, which can be rewritten as
\begin{equation}\label{key}
Q_1Q_1^T-Q_1^TQ_1 + Q_2Q_2^T-Q_2^TQ_2 = Q_2^TQ_1-Q_1Q_2^T+Q_1^TQ_2-Q_2Q_1^T
\end{equation}
Left multiplying $ Q_1 $ of the equality of $ \text{Re}{Q_1}=\text{Re}{Q_2} $ yields
\begin{equation}\label{eq:a}
 Q_1Q_1 + Q_1Q_1^T = Q_1Q_2 + Q_1Q_2^T
\end{equation}
Additionally, right multiplying $ Q_1 $ yields
\begin{equation}\label{eq:b}
Q_1Q_1 + Q_1^TQ_1 = Q_2Q_1 + Q_2^TQ_1
\end{equation}
Subtracting Eq. (\ref{eq:b}) from Eq. (\ref{eq:a}) results in
\begin{equation}\label{eq:c}
Q_1Q_1^T - Q_1^TQ_1  = Q_1Q_2 -Q_2Q_1 + Q_1Q_2^T-Q_2^TQ_1
\end{equation}
Analogously, left and right multiplying $ Q_2 $ establishes
\begin{equation}\label{eq:d}
Q_2Q_2^T - Q_2^TQ_2 = Q_2Q_1 - Q_1Q_2  + Q_2Q_1^T-Q_1^TQ_2
\end{equation} 
Substituting the Eq. (\ref{eq:c}-\ref{eq:d}) into the definition of normality leads to
\begin{equation}
Q_1Q_2^T-Q_2^TQ_1+Q_2Q_1^T-Q_1^TQ_2 = 0
\end{equation}
from which we have that $ Q_1Q_2^T-Q_2^TQ_1 $ is a skew symmetric matrix and the real part of all eigenvalues are $ 0 $. Minimizing the normality level of the integrated system can be quantified and analyzed by the minimization of Eq. (\ref{eq:min_normality_integrated}).

\section*{Acknowledgments}
X.W. acknowledges the project 62003156 supported by NSFC and project ``PCL Future Greater-Bay Area Network Facilities for Large-scale Experiments and Applications (LZC0019)". A.T. acknowledges partial support from NSF Grants EAR‐181190 and the UK Research and Innovation Global Challenges Research Fund Living Deltas Hub Grant NES0089261. Y.M. acknowledges partial support from the Government of Arag\'on, Spain through a grant to the group FENOL (E36-20R), by MINECO and FEDER funds (grant FIS2017-87519-P) and by Intesa Sanpaolo Innovation Center. The funders had no role in study design, data collection and analysis, or preparation of the manuscript.

\bibliographystyle{unsrt}
%\bibliography{Bib/biblio}

\begin{thebibliography}{10}

\bibitem{kivela2014multilayer}
Mikko Kivel{\"a}, Alex Arenas, Marc Barthelemy, James~P Gleeson, Yamir Moreno,
  and Mason~A Porter.
\newblock Multilayer networks.
\newblock {\em Journal of complex networks}, 2(3):203--271, 2014.

\bibitem{aleta2019multilayer}
Alberto Aleta and Yamir Moreno.
\newblock Multilayer networks in a nutshell.
\newblock {\em Annual Review of Condensed Matter Physics}, 10:45--62, 2019.

\bibitem{Myers2012}
Seth~A. Myers, Chenguang Zhu, and Jure Leskovec.
\newblock Information diffusion and external influence in networks.
\newblock In {\em Proceedings of the 18th ACM SIGKDD International Conference
  on Knowledge Discovery and Data Mining}, KDD '12, pages 33--41, New York, NY,
  USA, 2012. ACM.

\bibitem{vicsek1995novel}
Tam{\'a}s Vicsek, Andr{\'a}s Czir{\'o}k, Eshel Ben-Jacob, Inon Cohen, and Ofer
  Shochet.
\newblock Novel type of phase transition in a system of self-driven particles.
\newblock {\em Physical review letters}, 75(6):1226, 1995.

\bibitem{berman2009optimized}
Spring Berman, {\'A}d{\'a}m Hal{\'a}sz, M~Ani Hsieh, and Vijay Kumar.
\newblock Optimized stochastic policies for task allocation in swarms of
  robots.
\newblock {\em IEEE Transactions on Robotics}, 25(4):927--937, 2009.

\bibitem{prorok2017impact}
Amanda Prorok, M~Ani Hsieh, and Vijay Kumar.
\newblock The impact of diversity on optimal control policies for heterogeneous
  robot swarms.
\newblock {\em IEEE Transactions on Robotics}, 33(2):346--358, 2017.

\bibitem{saber2003consensus}
Reza~Olfati Saber and Richard~M Murray.
\newblock Consensus protocols for networks of dynamic agents.
\newblock 2:951--956, 2003.

\bibitem{abdelnour2014network}
Farras Abdelnour, Henning~U Voss, and Ashish Raj.
\newblock Network diffusion accurately models the relationship between
  structural and functional brain connectivity networks.
\newblock {\em Neuroimage}, 90:335--347, 2014.

\bibitem{de2016physics}
Manlio De~Domenico, Clara Granell, Mason~A Porter, and Alex Arenas.
\newblock The physics of spreading processes in multilayer networks.
\newblock {\em Nature Physics}, 12(10):901--906, 2016.

\bibitem{de2018fundamentals}
Guilherme~Ferraz de~Arruda, Francisco~A Rodrigues, and Yamir Moreno.
\newblock Fundamentals of spreading processes in single and multilayer complex
  networks.
\newblock {\em Physics Reports}, 756:1--59, 2018.

\bibitem{sahneh2015exact}
Faryad~Darabi Sahneh, Caterina Scoglio, and Piet Van~Mieghem.
\newblock Exact coupling threshold for structural transition reveals
  diversified behaviors in interconnected networks.
\newblock {\em Physical Review E}, 92(4):040801, 2015.

\bibitem{del2016synchronization}
Charo~I del Genio, Jes{\'u}s G{\'o}mez-Garde{\~n}es, Ivan Bonamassa, and
  Stefano Boccaletti.
\newblock Synchronization in networks with multiple interaction layers.
\newblock {\em Science Advances}, 2(11):e1601679, 2016.

\bibitem{cozzo2019layer}
Emanuele Cozzo, Guilherme~Ferraz de~Arruda, Francisco~A Rodrigues, and Yamir
  Moreno.
\newblock Layer degradation triggers an abrupt structural transition in
  multiplex networks.
\newblock {\em Physical Review E}, 100(1):012313, 2019.

\bibitem{gomez2013diffusion}
Sergio Gomez, Albert Diaz-Guilera, Jesus Gomez-Gardenes, Conrad~J
  Perez-Vicente, Yamir Moreno, and Alex Arenas.
\newblock Diffusion dynamics on multiplex networks.
\newblock {\em Physical review letters}, 110(2):028701, 2013.

\bibitem{sole2013spectral}
Albert Sole-Ribalta, Manlio De~Domenico, Nikos~E Kouvaris, Albert Diaz-Guilera,
  Sergio Gomez, and Alex Arenas.
\newblock Spectral properties of the laplacian of multiplex networks.
\newblock {\em Physical Review E}, 88(3):032807, 2013.

\bibitem{tejedor2018diffusion}
Alejandro Tejedor, Anthony Longjas, Efi Foufoula-Georgiou, Tryphon~T Georgiou,
  and Yamir Moreno.
\newblock Diffusion dynamics and optimal coupling in multiplex networks with
  directed layers.
\newblock {\em Physical Review X}, 8(3):031071, 2018.

\bibitem{cencetti2019diffusive}
Giulia Cencetti and Federico Battiston.
\newblock Diffusive behavior of multiplex networks.
\newblock {\em New Journal of Physics}, 21(3):035006, 2019.

\bibitem{hecker2009gene}
Michael Hecker, Sandro Lambeck, Susanne Toepfer, Eugene Van~Someren, and
  Reinhard Guthke.
\newblock Gene regulatory network inference: data integration in dynamic
  models—a review.
\newblock {\em Biosystems}, 96(1):86--103, 2009.

\bibitem{wang2019directionality}
Xiangrong Wang, Alberto Aleta, Dan Lu, and Yamir Moreno.
\newblock Directionality reduces the impact of epidemics in multilayer
  networks.
\newblock {\em New Journal of Physics}, 21(9):093026, 2019.

\bibitem{zhang2018altering}
Xizhe Zhang.
\newblock Altering indispensable proteins in controlling directed human protein
  interaction network.
\newblock {\em IEEE/ACM Transactions on Computational Biology and
  Bioinformatics}, 15(6):2074--2078, 2018.

\bibitem{wang2019structural}
Xiangrong Wang, Robert~E Kooij, Yamir Moreno, and Piet Van~Mieghem.
\newblock Structural transition in interdependent networks with regular
  interconnections.
\newblock {\em Physical Review E}, 99(1):012311, 2019.

\bibitem{cauchy1829equationa}
Augustin-Louis Cauchy.
\newblock Sur l’{\'e}quationa l’aide de laquelle on d{\'e}termine les
  in{\'e}galit{\'e}s s{\'e}culaires des mouvements des planetes.
\newblock {\em Exer. de math}, 4(1829):174--195, 1829.

\bibitem{thompson1968principal}
R~C Thompson.
\newblock Principal submatrices v: Some results concerning principal
  submatrices of arbitrary matrices.
\newblock {\em J. Res. Nat. Bur. Standards Sect. B}, 72(2):115--125, 1968.

\bibitem{sherman2013principally}
Michael~D Sherman and Ronald~L Smith.
\newblock Principally normal matrices.
\newblock {\em Linear Algebra and Its Applications}, 438(5):2617--2627, 2013.

\bibitem{fan1957imbedding}
Ky~Fan and Gordon Pall.
\newblock Imbedding conditions for hermitian and normal matrices.
\newblock {\em Canadian Journal of Mathematics}, 9:298--304, 1957.

\bibitem{Fan1950On}
Ky~Fan.
\newblock On a theorem of weyl concerning eigenvalues of linear
  transformations.
\newblock {\em Proceedings of the National Academy of Sciences}, 36(1):31--35,
  1950.

\bibitem{hiriart2007potpourri}
Jean-Baptiste Hiriart-Urruty.
\newblock Potpourri of conjectures and open questions in nonlinear analysis and
  optimization.
\newblock {\em SIAM review}, 49(2):255--273, 2007.

\bibitem{jiang2016simultaneous}
Rujun Jiang and Duan Li.
\newblock Simultaneous diagonalization of matrices and its applications in
  quadratically constrained quadratic programming.
\newblock {\em SIAM Journal on Optimization}, 26(3):1649--1668, 2016.

\bibitem{stark2006biogrid}
Chris Stark, Bobby-Joe Breitkreutz, Teresa Reguly, Lorrie Boucher, Ashton
  Breitkreutz, and Mike Tyers.
\newblock {BioGRID}: a general repository for interaction datasets.
\newblock {\em Nucleic acids research}, 34(1):D535--D539, 2006.

\bibitem{de2015muxviz}
Manlio De~Domenico, Mason~A Porter, and Alex Arenas.
\newblock {MuxViz}: a tool for multilayer analysis and visualization of
  networks.
\newblock {\em Journal of Complex Networks}, 3(2):159--176, 2015.

\bibitem{vickers1981representing}
M~Vickers and S~Chan.
\newblock Representing classroom social structure.
\newblock {\em Victoria Institute of Secondary Education, Melbourne}, 1981.

\end{thebibliography}

\end{document}